\documentclass{ecai}
\usepackage{graphicx}
\usepackage{latexsym}

\usepackage{bbm}
\usepackage{mathtools}
\usepackage{tikz}

\usetikzlibrary{arrows.meta}
\usetikzlibrary{calc,shapes.callouts,shapes.arrows}

\usepackage{amsmath,amssymb,amsthm}
\usepackage[capitalise]{cleveref}
\newtheorem{thm}{Theorem}

\newtheorem{lemma}[thm]{Lemma}

\newtheorem{mechanism}[thm]{Mechanism}
\newtheorem{contract}[thm]{Contract}

\crefname{thm}{Theorem}{Theorems}

\newcommand{\E}{\mathbb{E}}

\newcommand{\negl}{\textsf{negl}}

\ecaisubmission   
\begin{document}

\begin{frontmatter}

\title{Stackelberg Attacks on Auctions and Blockchain Transaction Fee Mechanisms}

\author[A]{\fnms{Daji}~\snm{Landis}\orcid{0000-0002-9985-0552}}
\author[B]{\fnms{Nikolaj I.}~\snm{Schwartzbach}\orcid{0000-0002-0610-4455}}
\address[A]{Bocconi University}
\address[B]{Department of Computer Science, Aarhus University}

\begin{abstract}
We study an auction with $m$ identical items in a context where $n$ agents can arbitrarily commit to strategies. In general, such commitments non-trivially change the equilibria by inducing a metagame of choosing which strategies to commit to. In this model, we demonstrate a strategy that an attacker may commit to that ensures they receive one such item for free, while forcing the remaining agents to enter into a lottery for the remaining items (albeit for free). The attack is thus detrimental to the auctioneer who loses most of their revenue. For various types of auctions that are not too congested, we show that the strategy works as long as the agents have valuations that are somewhat concentrated. In this case, all agents will voluntarily cooperate with the attacker to enter into the lottery, because doing so gives them a chance of receiving a free item that would have otherwise cost an amount commensurate with their valuation. The attack is robust to a large constant fraction of the agents being either oblivious to the attack or having exceptionally high valuations (thus reluctant to enter into the lottery). For these agents, the attacker may coerce them into cooperating by promising them a free item rather than entering in to the lottery. We show that the conditions for the attack to work hold with high probability when (1) the auction is not too congested, and (2) the valuations are sampled i.i.d. from either a uniform distribution or a Pareto distribution. The attack works for first-price auctions, second-price auctions and the transaction fee mechanism EIP-1559 used by the Ethereum blockchain.

The problem we study is natural in Web3 systems where agents natively interact using a blockchain. Thus, the agents are capable of deploying smart contracts that commit them to placing certain bids. In particular, the setting of an auction with multiple identical items models the transaction fee mechanisms that are used by blockchains to determine which transactions to include in the next block. Our work demonstrates that these mechanisms, in theory, are vulnerable to these attacks and may be cause for re-evaluation of the use of auctions in transaction fee mechanisms, at least when the networks are not too congested.
\end{abstract}

\end{frontmatter}

\section{Introduction}
Consider $n$ agents participating in an auction with $m$ copies of the same item. Each agent $i$ receives utility $v_i>0$ by obtaining one of the copies. Assume that all $v_i$ are distinct and ordered $v_1 < v_2 < \cdots < v_n$. Each agent places a bid $b_i \geq 0$ and the $m$ agents with the highest bids receive a copy of the item, at the cost of paying some function of the bids. If there are multiple agents with the same bid, the mechanism chooses uniformly at random between these agents. If $m\geq n$ then all agents receive a copy of the item, in which case the optimal strategy for each agent is to bid $b_i=0$. Thus, we will assume that $n = (1+\alpha)\,m$ for some \emph{congestion constant} $\alpha > 0$. 

In a first-price auction, an agent pays their own bid which results in untruthful behavior: it is well-known that the best response for an agent $i$ is to slightly outbid agent $n-m$ if their valuation exceeds this bid. That is, agent $i$ will place the following bid.
\begin{equation}
    b_i = 
    \begin{cases}
        v_{n-m} + \varepsilon&\text{if $i>{n-m}$,}\\
        0 & \text{if $i \leq n-m$.}
    \end{cases}
\end{equation}
Where $\varepsilon>0$ is some small constant, representing a negligible amount of money. It is not hard to see that this bidding strategy is indeed an equilibrium (at least up to $\varepsilon$). Of course, this requires the parties to be able to estimate the valuations of other parties. In some applications, this might not be a realistic assumption. Instead, the mechanism can be made truthful by letting each party with a winning bid pay $b_{n-m}$, a second-price\footnote{Technically, the auction should be called a $(n-m)^\text{th}$-price auction, or a Vickrey auction; we stick to second-price for simplicity.} auction \cite{vickrey}. In this case, it can be shown that the proposed mechanism is truthful so that each party will bid their valuations \cite{vickrey,clarke,groves}. While truthfulness is a desirable property, these auctions may be vulnerable to collusion \cite{roughgarden_eip1559}. 

\paragraph{Blockchains.} The auction described is also known as a \emph{transaction fee mechanism} and is used in blockchains to determine which transactions to include in the next block of data to include in the chain \cite{chung_shi_dream_mechanism}. Here, all pending transactions are public so it is reasonable to assume agents know the valuations of other parties. Blockchains canonically store transactions of cryptocurrency between different accounts \cite{chaum1,chaum2,bitcoin}, though many blockchains have since generalized this to support arbitrary execution of code, so-called smart contracts \cite{ethereum}. Smart contracts are decentralized programs that run on a virtual machine implemented by the blockchain. A smart contract maintains state, can transfer funds between parties, and responds to queries. A smart contract is guaranteed to be faithful to its implementation by security of the underlying blockchain \cite{ouroboros,kachina}. 

\paragraph{Stackelberg Equilibria.} It is well-known that being to commit to strategies, in general, changes the equilibria of the game by allowing an agent to commit to acting irrationally in some subgame, thus changing the equilibria of the game. The case with one agent being allowed to commit to strategies is known as a Stackelberg equilibrium and were first used in economics to model competing firms where one firm (the leader) has market dominance \cite{stackelberg}. This was later generalized to scenario where the leader commits to a strategy that depends on the strategies chosen by the other players, in what is known as reverse Stackelberg equilibria \cite{reverse_stackelberg,reverse_stackelberg_2}. This was further generalized by Hall-Andersen and Schwartzbach~\cite{smart_contracts} who consider a model of `universal commitments' where all players have smart contracts that are allowed to depend on each other sequentially. They show that this constitutes a hierarchy of equilibria that generalizes Stackelberg equilibria and reverse Stackelberg equilibria.

In this paper, we study transaction fee mechanisms involving agents who can universally commit to strategies using e.g. smart contracts. We call such attacks `Stackelberg attacks' and ask the following natural question.
\begin{quote}
    \em How do universal commitment to strategies impact the equilibria of transaction fee mechanisms?
\end{quote}
We show that these commitments drastically change the structure of equilibria for various types of auctions, thus opening for a Stackelberg attack wherein the buyers spontaneously organize to conspire against the auctioneer. In the attack, some agent commits to a strategy that ensures that they receive one of the items for free, while the remaining agents enter into a lottery for the remaining space on the block. The attack benefits all the buyers but is detrimental to the auctioneer who stands to lose most of their revenue. Note that while blockchains and smart contracts provide a natural setting in which to study these attacks, in principle the same framework can be used to analyze any setting in which agents can credibly commit to strategies, e.g. through reputation or by staking money. Understanding these attacks may also be important in predicting the behavior of advanced intelligent systems that have access to the internet (hence access to a blockchain).

\subsection{Our Results}
We demonstrate the existence of a Stackelberg attack on the transaction fee mechanism EIP-1559, which is used by Ethereum. This mechanism is a generalization of first-price auctions intended to fix various problems with first-price auctions in the context of transaction fee mechanisms \cite{roughgarden_eip1559}. By corollary, we show an attack on first-price auctions, which are used as transaction fee mechanisms in most other blockchains. The attack allows any agent to ensure they receive a copy of the item for free, while forcing (most of) the other agents to participate in a lottery for the remaining space. The attack works as long as the valuations are concentrated, in the sense that the largest values are not too much larger than the middle values. In this case, each agent voluntarily chooses the lottery because doing so will award them the item for free at some cost, while in the auction they would have to pay an amount commensurate with their valuation. If instead the valuations were spread out, the agents with a high valuation would not participate because they would be getting the item for a price much lower than their valuation.
\begin{thm}[Informal]\label{thm_intro_2}
    Let $v_1 < v_2 < \cdots < v_n$ be the valuations of the agents and $n>m$. Then if for some $k<m$, it holds that,
    $$
        \frac{v_{n-k+1}}{v_{n-m}} < \frac{n-k}{n-m},
    $$
    then EIP-1559 (including first-price auctions), as well as second-price auctions, are not side contract resilient.
\end{thm}
This is shown by explicitly demonstrating a strategy that an agent may commit to for which the equilibrium involves most parties entering into a lottery as described. The strategy extends also to second-price auctions.

We evaluate the economic efficiency of this new situation and show that, while the attack benefits all users, it is detrimental to the auctioneer. This impact on auctioneer suggests that successful and widespread deployment of the attack would be detrimental to the viability of running the auctions. Therefore, our analysis is grounds for reevaluation of the auctions for transaction fee mechanisms. Formally, we define the \emph{price of defiance} as the ratio between the utility an agent receives by cooperating versus the utility they would receive by deviating (or \emph{defying} the attacker). We give a probabilistic bound on the price of defiance for the attack.
\begin{thm}[Informal]
    Suppose $n$ agents participate in an auction with $m$ identical items, and $n = (1+\alpha)\,m$ for some $\alpha > 0$. If the agents have valuations that are i.i.d. uniform, then with high probability, the price of defiance is at least $1+\alpha$.
\end{thm}
 We show that the conditions required to apply \cref{thm_intro_2} are natural, in the sense that they are satisfied with high probability at certain levels of congestion when the valuations are sampled from two natural distributions.
\begin{thm}[Informal]
    Suppose $n$ agents participate in an auction of $m$ identical items, and $n = (1+\alpha)\,m$ for some $\alpha>0$. Then the conditions required for \cref{thm_intro_2} to apply hold with overwhelming probability if either of the following two conditions are satisfied.
    \begin{enumerate}
        \item The valuations are sampled i.i.d. from a uniform distribution and, 
        $$
            0 \leq \alpha < 0.53.
        $$
        \item The valuations are sampled i.i.d. from a Pareto distribution with parameter $p>1$ and $0 \leq \alpha < \alpha(p)$ for some function $\alpha$ with, 
        $$
            \lim_{p \rightarrow \infty} \alpha(p) \approx 0.69.
        $$
    \end{enumerate}
\end{thm}

Our work highlights the difficulty in designing smart contracts and suggests that other smart contracts that have already been deployed on major blockchains may be susceptible to Stackelberg attacks.

\subsection{Related Work}
Stackelberg equilibria are quite well-studied and are important e.g. in control theory \cite{basar1979closed,bloem2010stackelberg,roughgarden_stackelberg,doan2021peer} and security games \cite{korzhyk2011stackelberg,kar2017trends,sinha2018stackelberg}. In general, these commitments change the equilibria in highly non-trivial ways \cite{amir_grilo_1999,sherali1983stackelberg}, and they are known to be hard to compute in general \cite{stackelberg_complexity,korzhyk2010complexity,bai2021sample}, though there are some games for which the Stackelberg equilibrium can be shown to coincide with the subgame perfect equilibrium (SPE) \cite{basu1995stackelberg}. Variants of Stackelberg equilibria are known for some auction scenarios, e.g. for all-pay auctions with complete information \cite{konrad2007generalized}, for procurement auctions \cite{garg2005design,garg2008mechanism}, and for repeated auctions \cite{nedelec2020robust}. Reverse Stackelberg equilibria are less studied, however they also find applications in routing \cite{groot2014toward}, in control theory \cite{groot2017hierarchical,tajeddini2020decentralized}, and sparsely in auctions \cite{nedelec2019adversarial} though in a different context than what we consider in this work. In fact, the attack we consider in this work only works for $n \geq 3$ which means it inherently eludes analysis as a (reverse) Stackelberg equilibrium. Little is known of the generalizations of Stackelberg equilibria that we study in this work, aside from the complexity results shown in \cite{smart_contracts}. 

In recent years, there has been increased interest in analyzing blockchain transaction fee mechanisms using techniques from classic mechanism design. A line of work, \cite{btc_fee_markets,roughgarden_eip1559}, identifies three desiderata of such mechanisms:
\begin{enumerate}
    \item \emph{user-incentive compatibility (UIC)}. The users are incentivized to bid truthfully;
    \item \emph{miner-incentive compatibility (MIC)}. The miners are incentivized the implement the mechanism as prescribed;
    \item \emph{off-chain agreement proofness (OCA proofness)}. No coalition of miners and users can increase their joint utility by deviating from the mechanism. 
\end{enumerate}
In \cite{roughgarden_eip1559}, Roughgarden shows that EIP-1559 satisfies MIC and OCA proofness when the block size is large and shows that it is not UIC, in the sense that users may benefit by bidding strategically. Here, OCA proofness means that the users and the miner cannot benefit by agreeing to off-chain payments and thus captures a specific type of commitment to strategies. Chung and Shi~\cite{chung_shi_dream_mechanism} show that no mechanism can simultaneously be UIC and 1-OCA proof. These results are shown in a model where agents cannot universally commit to strategies, and indeed we show that, arguably, neither of these three properties hold in a model where the agent can universally commit to strategies.
\section{A Stackelberg Attack on Auctions}
\label{sec:attack}

We will consider a set of $n$ transactions competing for space on a block of size $m$. Each transaction is assumed to be owned by exactly one agent that we identify with the integers $\{1,2,\ldots,n\}$. Each agent $i$ has a valuation $v_i>0$ of their transaction, which is the utility they gain by having their transaction included in the block for free. We assume the agents are rational, risk-neutral, and have a quasilinear utility functions. We will typically take each $v_i$ as sampled i.i.d. from some known underlying distribution ${D}$. It will be convenient to assume that agents know each others' valuations precisely, i.e. we assume the values $v_1, v_2, \ldots, v_n$ are public and known to all the agents. Although this assumption is false in practice, by fixing ${D}$, the parties can mostly infer the valuations of the other parties, as these values will be highly concentrated around their expectations, if the number of agents is sufficiently large. This approach is used in practice on Ethereum, where several services provide tip estimations based on the current network congestion \cite{donmez2022transaction}. 

We assume each agent is capable of deploying a smart contract that may bid on their behalf, conditioned also on the smart contracts deployed by the other agents. To formalize this, we may use the model of \cite{smart_contracts}: fix some extensive-form representation of the sealed-bid auction, e.g. (1) choose an arbitrary order of the agents, (2) construct the $n$-horizon game with the agents in the specified order with each layer having a subgame corresponding to each bid that a specified agent may place, (3) add information sets to ensure agents are not aware of the bids made by the other agents, (4) add utility vectors corresponding to the type of auction (first-price, second-price, etc.). Finally, add `smart contract moves' to the top of the game tree for each player. These moves are special nodes that are syntactic sugar for the larger `expanded tree' that results from computing all appropriate cuts in the game tree and reattaching them with a node belonging to that player. By expanding these moves in a bottom-up fashion, this gives a natural way for contracts to condition on the contracts deployed by other agents and is shown to generalize (reverse) Stackelberg equilibria. For more details, we refer to \cite{smart_contracts}, though we trust that the intuitive understanding of `contracts that depend on other contracts' suffices for the purposes of this work. An auction that is weakly strategically equivalent (i.e. the equilibrium payoffs are equal) to itself with smart contract moves is said to be \emph{Stackelberg resilient}. 

We now give our model of the transaction fee mechanism EIP-1559 used by Ethereum since 2021\footnote{In practice, the block size of EIP-1559 is variable and we shall let $m$ denote its maximum possible value. In practice, the base fee would be adjusted to ensure that $\E[n]=m/2$, however the case of $n \leq m$ is not interesting (as all transactions will simply be included) so we take $m$ to be the maximum value and assume $n>m$.}. It generalizes first-price auctions by including a \emph{base fee} $B \geq 0$ that each agent has to pay that is burned. The base fee is continuously adjusted by the network to balance the demand to ensure each block is half full (in expectation). A first-price auction with $m$ identical items is retained as a special-case when $B=0$.

\begin{mechanism}\hspace{-1mm}\textsc{\textup{(EIP-1559)}}.\label{alg:eip-1559}
\begin{enumerate}
    \item[1.] Each party $i\in[n]$ submits a transaction of value $v_i>0$ and makes a deposit of $B+\tau_i$ funds where $\tau_i\geq0$ is an optional tip.
    \item[2.] A miner finds a block, and selects a $T\subseteq [n]$ with $|T|=m$ that maximizes $\sum_{i \in T} \tau_i$. If there are multiple such $T$'s, it selects $T$ uniformly at random from all suitable sets.
    \item[3.] Each party $i \in T$ has their transactions included in the block and loses their deposit, in total gaining $v_i - B - \tau_i$ money; each party $j\not\in T$ is returned their deposit of $B+\tau_j$ money and gains 0.
    \item[4.] The miner receives $\sum_{i \in T} \tau_i$ money.
    \item[5.] The network adjusts the base fee $B$ depending on $m$ and $n$.
\end{enumerate}
\end{mechanism}

Keeping in tune with auction terminology, moving forward we will refer to the miner as the \emph{auctioneer}. As per the introduction, we will let $n=(1+\alpha)\,m$ for some \emph{congestion constant} $\alpha>0$. Let $\varepsilon>0$ be the smallest unit of currency, and assume it is sufficiently small, i.e. $\varepsilon\ll v_i$, to mostly be ignored in calculations. In practice, on Ethereum, as of 2022, we have $\varepsilon \approx \$10^{-12}$. 

We now propose a Stackelberg attack on \cref{alg:eip-1559}: essentially, the leading contract agent commits to paying $2\varepsilon$, conditioned on everyone else committing to bidding $\varepsilon$. In this case, the leading contract agent has their transaction included at essentially zero cost, while everyone else enters into a lottery. If anyone does not comply, the leading contract agent instead submits the bid they would have submitted without the contracts, or one slightly higher. This forces each other agent to decide between a lottery and a first-price auction. We will show that when the valuations of the transactions are somewhat concentrated, the agents prefer the lottery over the first-price auction, as they would otherwise have to pay a bid commensurate with their valuation, while in the auction they may receive the item for free.

As a warm-up and ongoing example, we look at the case where there are three agents and two slots up for auction, that is $n=3$ and $m=2$, to illustrate the attack. This models a case where there are three buyers that wish to purchase two identical items --- we may imagine these to be exchanges that control large quantities of user transactions, such as e.g. Coinbase or Binance \cite{alexander2022role}. Note that in this example we have $\alpha=\frac12$. Suppose that agents $1,2,3$ have valuations $0<v_1 < v_2 < v_3$, respectively. In a first price auction, where the valuations of the respective parties are known, the $m$ agents with the highest valuations only need to outbid the agent with $m+1$ highest valuation, who is unwilling to bid beyond their valuation and receive negative utility. In our example, agents 2 and 3 will bid slightly higher than the valuation of agent 1, yielding the following utilities:
$u_1 = 0$, $u_2 = v_2 - v_1 - \varepsilon$, $u_3 = v_3 - v_1 - \varepsilon$.

We now equip these three agents with contracts.  If the agent with the leading contract can make a credible and enforceable threat with the contract, they may force other agents to accept the lottery at the price $\varepsilon$, thereby guaranteeing the leading agent space an item at price of 
 $2\varepsilon$. The viability of such a threat depends on the agents' valuations. Agents will only comply if their expected utility is higher when they cooperate compared to when the threat is executed.

Consider first the case when agent $3$ is the leading contract agent. The contract will commit agent $3$ to bidding either $2\varepsilon$, if the two other agents commit to playing $\varepsilon$, or to bidding the usual first price bid of $v_1+\varepsilon$ otherwise.  If the contract works, agent $3$ enjoys utility $v_3 - 2 \varepsilon$, a better result than the first price auction utility of $v_3 - v_1 - \varepsilon$. The desirable outcome is also clear for agent 1: the lottery case yields utility $\frac{1}{2}(v_1 - \varepsilon)$, which is better than the first price auction utility of $0$. Therefore, both $1$ and $3$ will submit to the contract. Agent 2 will cooperate if the first price utility is lower than the lottery utility, that is if $v_2-v_1 -\varepsilon < \frac{1}{2}(v_2-\varepsilon)$, which reduces to $v_1 + \frac{1}{2} \varepsilon > \frac{1}{2}v_2$. The attack would not work if the valuations were less concentrated. If agent 2 is the lead contract holder, the attack works if $v_1 + \frac{1}{2} \varepsilon > \frac{1}{2}v_3$, a more stringent concentration requirement.  If agent 1 has the leading contract, they may threaten to bid $v_2 + \varepsilon$, knowing they will likely not have to pay it. In this scenario, agent 1 has a credible threat if $v_2 + \frac{1}{2} \varepsilon > \frac{1}{2}v_3$, similar to the conditions for agent 3.

The attack generalizes readily to a larger number of agents, although the requirement on the valuations becomes stronger with more agents. In particular, the attack no longer works if even a single agent has a valuation that is significantly higher than the median. However, the leading contract agent may persuade such agents into participating by promising them a free item from the auction, taking some of the spots intended for the lottery. We denote by $C \subseteq [n]$ the \emph{coalition} of agents (with $|C|=k$ for some $k<m$) who are given free items. This significantly loosens the valuation requirement and allows us to show that the attack works even if $k<m$ of the parties have large valuations. The set $C$ may also be used to capture those agents who are oblivious to the attack, thus modeling the (very realistic) scenario where some of the agents are not aware of the attack and cannot respond accordingly. We have not explicitly accounted for this; doing so would give a slightly stronger bound in the following but would not fundamentally change the analysis. We now describe the attack in more detail.

\begin{thm}\label{thm:eip_not_resilient}
    Consider $m$ identical items, and let $\varepsilon\ll v_1 < v_2 < \cdots < v_n$ be the valuations of the $n$ buyers, with $n=(1+\alpha)\,m$ for some $\alpha>0$. If for some $k<m$ it holds that, 
    \begin{equation}\label{eq:valuation_bound}
        \frac{v_{n-k+1}-B}{v_{n-m}} < \frac{n-k}{n-m},
    \end{equation}
    then EIP-1559 is not Stackelberg attack resilient.
\end{thm}
\begin{proof}
    Assume that each agent has exactly one transaction, and let agent $i$ be the agent associated with the transaction of valuation $v_i$. Suppose the contract agents are ordered $i_1, i_2, \ldots, i_n$, where $i_1$ is the leading contract agent. Now consider the following contract $A^C_u$, parameterized by an integer $u\in[n]$ that represents the index of the contract order and a set $C \subseteq [n]$ with $i_1 \in C$ and $|C|=k$ for some $k\leq m$.
    \begin{contract}{\hspace{-1mm}\normalfont($A^C_{u}$).}\hfill\label{attack}
        \begin{enumerate}
            \item[1.] If $u=n$, play $\varepsilon$.
        \item[2.] If $u<n$, play $v_{n-m}+\varepsilon$ if $v_{i_u} > v_{n-m}+\varepsilon$ and 0 otherwise in every subgame where any agent $i_{v}$ with $v>u$ does not play the contract $A^C_{v}$; otherwise play $2\varepsilon$ if $u\in C$, and $\varepsilon$ if $u\not\in C$.
        \end{enumerate}
    \end{contract}
   Now suppose the leading contract agent deploys the contract $A^C_{1}$ with $|C|=k<m$ and $i_1 \in C$. If they are successful, their transaction will be added with certainty for a cost of $2\varepsilon$, thus gaining $v_{i_1}-2\varepsilon$. Consider the strategy of agent $j$ when every other agent plays \cref{attack}. If $j \in C$, then clearly for small $\varepsilon$, agent $j$ will comply with the threat. If instead $j \not\in C$, they will play \cref{attack} to obtain a value of $v_{j}-\varepsilon$ with probability $\frac{m-k}{n-k}$. If they do not play \cref{attack}, by design, all agents revert to a first-price auction. Then agent $j$ can either tip too little to win or tip at least $v_{n-m}+\varepsilon$ to have their transaction included. If $j\leq n-m$, this exceeds their valuation, hence they prefer \cref{attack}, as its expected payoff is $\frac{(m-k)(v_i - B - \varepsilon)}{n-k} > 0$. If instead $j > n-m$, they can choose not to comply with the threat to gain $v_{j} - v_{n-m} - B - 2\varepsilon$ utility. It follows that such an agent will comply with the threat if $v_j - v_{n-m} - B - \varepsilon > \frac{(m-k)(v_{j} - B - \varepsilon)}{n-k}$, which when ignoring $\varepsilon$'s, solves to $\frac{v_{j} - B}{v_{n-m}} < \frac{n-k}{n-m}$. But this is guaranteed to hold by \cref{eq:valuation_bound}, since $v_{j}\leq v_{n-k+1}$ for any $j$. Thus, complying with the threat is an equilibrium and hence EIP-1559 is not Stackelberg resilient. 
\end{proof}

Note that by letting $B=0$ we obtain a regular first-price auction, and hence \cref{thm:eip_not_resilient} implies that the transaction fee mechanisms of Ethereum, Bitcoin, and most other blockchains are not Stackelberg resilient, regardless of whether there is a base fee or not. We observe that the attack works also for second-price auctions.

\begin{thm}
    Consider a second-price auction again with $m$ identical items, and $n$ buyers, in keeping with \cref{thm:eip_not_resilient}.  If \cref{eq:valuation_bound} holds then the auction is not Stackelberg attack resilient.
\end{thm}
\begin{proof}[Proof (Sketch)]
    We consider the same attack, \cref{attack}. As we have seen, in the EIP-1559 setting, which is a first price auction when $B=0$, bidders have perfect information and must only bid just enough to outbid the $(n-m)^\text{th}$ highest one out with a bid of $v_{n-m}+\varepsilon$ and will be charged that same amount. In the second price auction, they can bid their valuation or stick with $v_{n-m}+\varepsilon$.  In any case, if they are included, the agent will pay $v_{n-m}$, a slight discount on the $v_{n-m}+\varepsilon$ cost in the first price setting. Thus \cref{attack} can be used and the scenario in which the attack works will look the same. If the attack does not work and agents revert to the equilibrium as it would be without contracts, but this time with the slightly different cost.  Note that in the proof of \cref{thm:eip_not_resilient} we drop the epsilons that constitute the difference between the first and second price auctions.  So by the proof of \cref{thm:eip_not_resilient}, a the second price auction as described is also not Stackelberg resilient.    
\end{proof}

\paragraph{Risk Aversion.} It is natural to wonder if the attack will still work if agents are risk averse.  To model risk aversion, agents have some concave utility function $u=U(\cdot)$.  If, for example, an agent gets a slot for free at valuation $v_i$, their utility would be defined to be $u=U(v_i)$.  For $U(\cdot)$ to be concave, we must have $U((1-p)x + p y) \geq (1-p)U(x)+p U(y)$ where $(x,U(x))$ and $(y,U(y))$ are two points on the utility function and $p \in [0,1]$.  Graphically, this implies that any point on the line between $(x,U(x))$ and $(y,U(y))$ is on or below the utility function. This is the line tracing out the utility of the function of a coin toss with probability $p$ between $U(x)$ and $U(y)$. This models risk aversion because the utility of any outcome based on a coin toss between two outcomes will be on or below the curve, which represents the utility of outcomes that are certain. If we make the assumption that $x=U(x)=0$ and set $y=v_i$, we have $U(p v_i)\geq p \,U(v_i)$. Note that in the proof of \cref{eq:valuation_bound}, we required the condition, here simplified, that $v_i-v_{n-m-k+1}>p v_i$.  If we instead had some concave utility function, this would be $U(v_i-v_{n-m})>p U(v_i)$.  Given that $U(p v_i)\geq p U(v_i)$, the condition found in \cref{eq:valuation_bound} is necessary, but not necessarily sufficient, for the contact attack to be viable. Finding the exact condition requires $U(\cdot)$ to be known.

\section{Everyone Benefits Except for the Auctioneer}
\label{sec:econ}
In the following, we will assume for the sake of argument that $k=1$ and that $\varepsilon=0$. As $k$ increases, more agents with high valuations get free entry when $\varepsilon=0$. Thus their relatively high valuations are counted into social welfare.  As long as this elite group is relatively small, this will have little impact on the chances of the lottery players, meaning the allowance of a relatively small $k>1$ will increase social welfare. 

We define the \emph{price of defiance}, a ratio of sets of equilibrium related to the price of anarchy \cite{poa}. Let $S$ be the set of all strategy profiles in the game, and take two sets $C \subseteq S$, some set of strategies, and $E \subseteq S$, the set of equilibria of the game. We take the set $C$ to be the set of equilibria after a successful contract attack has been deployed. Define,
\begin{equation}\label{eq:pod}
    PoD = \frac{\max_{s \in C}\text{Welf}(s)}{\min_{s \in E} \text{Welf}(s)}.
\end{equation}
We look at the ratio between the best of a subset of possible outcomes and that same worst equilibrium. This differs from the price of anarchy in that we want to compare some subset of strategies, here those that become equilibria due to the introduction of a contract attack, rather the optimal solution, to the game's usual equilibria. We have  $PoD \leq PoA$. 

Our set $C$ is the set of equilibrium arising from agents having and complying with Contract 4.1. There are up to $n$ equilibria in the set, one for each choice of agent with leading contract. To analyze the price of defiance we will need concentration bounds on the valuations of the parties. Order the players with valuations $v_1 < v_2 < \ldots < v_n$, then $v_i \sim \text{Beta}(i,n+1-i)$. Say a function $f$ is \emph{negligible} if $f(x) = o(x^{c})$ for every constant $c \in \mathbb{R}$, i.e. if it grows slower than the inverse of any polynomial. We will make use of the following concentration bound on order statistics from the uniform distribution.
\begin{lemma}[Skorski, \cite{beta_bounds}]\label{lemma:beta_bound}
    Let $X \sim \text{Beta}(\alpha,\beta)$ for $\alpha,\beta>0$, and define,
    \begin{align*}
        v^2 = \frac{\alpha\beta}{(\alpha+\beta)^2(\alpha+\beta+2)}, && 
        c_0 = \frac{|\beta-\alpha|}{(\alpha+\beta)(\alpha+\beta+2)}.
    \end{align*}
    Then for any $\varepsilon>0$, it holds that,
    \begin{align*}
        \Pr\left[\left \lvert X-\E[X] \right \rvert > \varepsilon\right] &\leq 2\exp\left(-\frac{\varepsilon^2}{2v^2+2 \varepsilon \max\left\{v, c_0 \right\}}\right).
    \end{align*}
\end{lemma}
\begin{lemma}\label{lemma:bound}
    Let $X_1, X_2, \ldots, X_n \sim U[0,1]$, and let $X_{(1)} < X_{(2)} < \cdots < X_{(n)}$ be the $n$ order statistics. Then, 
    \begin{align*}
        \left \lvert X_{(i)} - \frac{i}{n+1} \right \rvert = \Tilde{O}(1/n), && \text{for every $i=1\ldots n$},
    \end{align*}
    except with negligible probability in $n$.
\end{lemma}
\begin{proof}
    We make use of \cref{lemma:beta_bound} to bound the error term and must therefore first find the relevant values of $v$ and $c_0$.  It is a fact that such order statistics have the distribution $\text{Beta}(i, n+1-i)$, i.e. $\alpha = i$ and $\beta = n+1-i$. Thus, for all values of $i$, we must have $\alpha + \beta = n+1$.  It is easy to see that we find the largest value $v^2$ from \cref{lemma:beta_bound} when $\alpha = \beta = \frac{n+1}{2}$.  This case yields 
\begin{align*}
    v^2 \leq \frac{\frac{n+1}{2}\frac{n+1}{2}}{\left(\frac{n+1}{2}+\frac{n+1}{2}\right)^2\left(\frac{n+1}{2}+\frac{n+1}{2}+2\right)}=\frac{1}{4(n+3)}.
\end{align*}
The value of $c_0$ is largest when the numerator is largest, which is clearly when $|\beta - \alpha| = n-1$. Note that this is a specifically different case from when $v^2$ is largest.  When we go on to find the error bounds on specific $v_i$'s we will refine the bound at this step. Thus, we have the following bounding value,
\begin{align*}
    c_0 \leq \frac{n-1}{(n+1)(n+3)}.
\end{align*}
It is easy to see that $c = \max \{v, c_0 \}=c_0$.
Thus we can write down the bound for any $i$,
\begin{align*}
     \Pr\left[\left\lvert X_{(i)} - \E[X_{(i)}]\right\rvert > \delta\right] 
        &< 2\exp\left(-\frac{\delta^2}{2v^2+2c\delta}\right)\\
        &\leq 2\exp\left(-\frac{\delta^2}{2\frac{1}{4(n+3)}+\frac{2\delta(n-1)}{(n+1)(n+3)}}\right)\\
        &= 2\exp\left(-\frac{\delta^22(n+3)(n+1)}{(n+1)+4\delta(n-1)}\right)\\
        &< 2\exp\left(-\frac{\delta^2 2n^2}{(n+1)+4\delta n}\right)\\
        &\approx 2\exp\left(-\frac{\delta^2 2n}{1+4\delta }\right)\\
        &=2 \exp (-\Omega (\delta n)).
\end{align*}
If we take $\delta=\frac{\log^2{n}}{n+1}=\Tilde{O}(1/n)$, we obtain the bound, 
\begin{align}
     \Pr\left[\left\lvert X_{(i)} - \E[X_{(i)}]\right\rvert > \delta\right] < 2 \exp(-\omega(\log n)),
\end{align}
which is negligible in $n$. We conclude by doing a union bound on all $n$ valuations.
\end{proof}

\begin{thm}
    For uniformly distributed valuations, the price of defiance is at least $1+\alpha-o(1)$ except with probability negligible in $n$.
\end{thm}
\begin{proof}
It is easy to see that the maximal choice $s \in  C$ is when the agent with the highest valuation has the contract. There is only one choice for equilibrium $s \in C$. Thus we have,
     \begin{align}
         PoD &= \frac{\max_{s \in C}\text{Welf}(s)}{\min_{s \in E} \text{Welf}(s)}
         = \frac{ \left( \sum_{j=1}^{n-1} \frac{m-1}{n-1}(v_j-\varepsilon) \right) +v_n -2\varepsilon }{\left( \sum_{i=n-m+1}^n v_i -v_{n-m} -\varepsilon \right)} \nonumber\\
         &\approx \frac{\frac{m-1}{n-1} \left( \sum_{j=1}^{n-1} v_j \right) +v_n }{\left( \sum_{i=n-m+1}^n v_i   \right)-mv_{n-m}}.\label{pod_approx}
     \end{align}
At this stage we have not yet used any assumptions on the distribution of the valuations. If the contract attack works, that is if the valuations are in keeping with the in condition from \cref{thm:eip_not_resilient}, we have $PoD >1$.  This can be seen mathematically by substituting the condition into the denominator of \cref{pod_approx} above.  Intuitively, given that the threat is just the usual first price auction when the contract holder is agent $n$, the other agents will acquiesce only if their utility would be higher in the lottery.  Thus total lottery welfare, the numerator, must be higher than the auction, the denominator, leading to a $PoD>1$ in the general case. Each $v_i$ is the $i^{th}$ order statistic of a uniformly distributed random variable, that is $v_i=X_{(i)}$ where $X_i$ is sampled i.i.d. from the uniform distribution on $[0,1]$. By linearity of expectation, we have that,
\begin{align*}
    \E \left[ \sum_{i=n-m+1}^n v_i \right] = \sum_{i=n-m+1}^n \frac{i}{n+1} 
    &= \frac{1}{n+1 }\left( \sum_{i=0}^n i - \sum_{k=0}^{n-m}k \right)\\
    &=\frac{n}{2}-\frac{(n-m)(n-m+1)}{2(n+1)}, \\\intertext{and that,}
    \E \left[\sum_{j=1}^{n-1}v_j \right]  = \sum_{j=1}^{n-1}\frac{j}{n+1} &= \frac{(n-1)n}{2(n+1)}.
\end{align*}
We proceed to lower bound $PoD$ using \cref{lemma:bound} to yield,
\begin{align*}
    PoD &\geq \frac{\frac{m-1}{n-1} \left( \frac{(n-1)n}{2(n+1)}-(n-1)\delta\right) +\frac{n}{n+1}-\delta}{\frac{n}{2}-\frac{(n-m)(n-m+1)}{2(n+1)}+m\delta -m\left( \frac{n-m}{n+1}+\delta \right)} \\ 
    &=\frac{n(m+1)-2m(n+1)\delta}{m(m+1)+4m (n+1)\delta}
\end{align*}
We now condition on the errors of the valuations being bounded by $\delta=\frac{(m+1) \log^2{n}}{2m(n+1)}$, which we know to happen except with negligible probability by \cref{lemma:bound}. Then we obtain the following bound,
\begin{align*}
    PoD &\geq \frac{n-\log^2(n)}{m+\log^2(n)} = 1+\alpha - o(1),
\end{align*}
as desired. 
\end{proof}
This arguably suggests that lotteries should be used for transaction mechanisms instead when the valuations can be believed to be of similar size. In the $n=3$, $m=2$ case, we have 
 \begin{align*}
     PoD = \frac{\frac{v_1}{2}+\frac{v_2}{2}+v_3-3\varepsilon}{v_2+v_3-2v_1}
 \end{align*}
If the condition for the contract attack working as discussed in the example in \cref{sec:attack} hold, that is, if $v_1+ \frac{1}{2}\varepsilon>\frac{1}{2}v_2$, the ratio becomes
\begin{align*}
    PoD > \frac{\frac{v_1}{2}+\frac{v_2}{2}+v_3-3\varepsilon}{v_3 +\varepsilon}
\end{align*}
which is clearly larger than one. 

It is important to note that while the attack benefits all the agents with transactions, it is detrimental to the auctioneer who lose essentially all of their revenue. Continuing with our $n=3$, $m=2$ example, we see the auctioneer will earn $2(v_1+\varepsilon)$, twice the twin winning bids from agents 2 and 3, in the auction case. If the contract attack is successfully executed, the auctioneer income will be $3\varepsilon$, $2\varepsilon$ from the leading contract holder, regardless of which agent this is, and $\varepsilon$ from the winner of the lottery.  Thus almost all the revenue is lost; the auctioneer will miss out on $2v_1-\varepsilon$ income. If there were a base fee and all agents had a valuation larger than said base fee, i.e. $v_1>B$, the first price revenue would be $2(v_1+\varepsilon-B)$. The lottery revenue will continue to be $3\varepsilon$ and the income lost to the attack will be $2(v_1-B)-\varepsilon$.

\section{The Attack Works for Natural Distributions}
\label{sec:distributions}
In this section, we show that the conditions required for the attack to work are satisfied with high probability under reasonable assumptions. For the sake of analysis, we will assume that $B=0$. The results obtained are qualitatively similar when the valuations are much larger than the base fee.

We continue with our illustration of the $n=3$, $m=2$ case, now assuming that the players have valuations that are uniformly distributed on $[0,1]$. As before, we have three valuations $v_1 < v_2 < v_3$ and we can now make use of the distribution.  The valuations in order are order statistics, that is $v_i = X_{(i)}$ where all $X_i$ are sampled i.i.d. from the uniform distribution on $[0,1]$. Using the well known fact that order statistics on the uniform distribution follow specific beta distributions, we get the following distributions and their expectations: $v_1 \sim \text{Beta}(1,3)$ 
 yielding, $\mathbb{E}[v_1]=\frac{1}{4}$; $v_2 \sim \text{Beta}(2,2)$, yielding $\mathbb{E}[v_2]=\frac{1}{2}$; and $v_3 \sim \text{Beta}(3,1)$, with
 $\mathbb{E}[v_3]=\frac{3}{4}$.
Note that the variance for all the distributions is $\text{Var}[v_i]\leq 1/20$ and we will not take it into account moving forward. In the first price auction scenario, we can see that if agents 2 and 3 bid just enough to outbid agent 1, i.e. $\frac{1}{4}+\varepsilon$, they will secure their slots as cheaply as possible. So in the first price auction the players will have the expected utilities
$\E[u_1] = 0$, $\E[u_2] = \frac{1}{4}-\varepsilon$, and $\E[u_3]  =\frac{1}{2}-\varepsilon.$

If agent 1 has the leading contract, they can threaten to outbid agent 2 with a bid of $\frac{1}{2}+\varepsilon$. If the threat were to be carried out, agent 2 would lose their slot and receive utility $0$ and agent 3, secure in the top spot, but now having to outbid agent 2, will receive $\frac{1}{4}-\varepsilon$.  If agents 2 and 3 comply with the threat, i.e. bid $\varepsilon$ and enter a lottery, they will have expected utilities $\frac{1}{4}-\frac{\varepsilon}{2}$ and $\frac{3}{8}-\frac{\varepsilon}{2}$, respectively. It is clear that these utilities are more desirable than ignoring the threat, and the attack can be executed. Agent 1 will enjoy expected utility $\frac{1}{4}-2\varepsilon$. Note that agents 1 and 2 have higher utility than they would have had in the first price auction, but agent 3 is hurt by the attack.

If agent 2 has the leading contract, their best attempt at a threat is outbidding agent 1 with a bid of $\frac{1}{4}+\varepsilon$. This is no threat at all as it simple coincides with their first price strategy. If instead agent 3 has the leading contract, we once again have a viable attack. Since agent 3 already outbids the others, their contract endowed strategy is more a proposition for mutual benefit than a greedy attack.  If the other two parties enter into a lottery at price $\varepsilon$ and agent 3 bids $2\varepsilon$ we have expected utilities $\E[u_1] = \frac{1}{8}-\frac{\varepsilon}{2}$,$\E[u_2] = \frac{1}{4}-\frac{\varepsilon}{2}$, and $\E[u_3]  =\frac{3}{4}-2\varepsilon.$ It can be easily seen that everyone benefits in this situation and the attack will work.
It is an easy calculation to find that $PoD\approx3/2$. Regardless of which agent has the leading contract, if the attack works, the total tip paid to the auctioneer will be $3\varepsilon$.  In the first price auction the expected auctioneer payout is $\frac{1}{2}+2\varepsilon$. The difference constitutes an almost complete loss of revenue. 
\begin{lemma}[Xu, Mei, Miao, \cite{ratio_order_stats}]\label{lemma:ratio_order_stats}
    Let $X_1, X_2, \ldots, X_n \sim U(0,1)$ be i.i.d. Let $i<j$ and define $R_{ij} = \frac{X_{(j)}}{X_{(i)}}$ and let $f(\cdot)$ be its density function with support $[1,\infty)$. Then for every $r\geq 1$,
    \begin{align*}
        f(r) &
        = \frac{n!(r-1)^{j-i-1}}{(i-1)!(j-i-1)!(n-j)!r^j} \,  \int_{0}^1 (1-u)^{j-1}\,u^{n-j} \, \mathrm{d}u.
    \end{align*}
\end{lemma}
\begin{thm}
    Suppose $n$ buyers participate in an auction of $m$ identical items where $n=(1+\alpha)\,m>m+1$. If the valuations of the items are sampled uniformly at random and, $0 \leq \alpha < 0.53$, then first-price auctions are not Stackelberg resilient, except with probability negligible in $m$.
\end{thm}
\begin{proof}
    We will show that \cref{eq:valuation_bound} holds except with probability $\negl(n)$. Suppose w.log. that the valuations are sampled uniformly from $[0,1]$ and let $v_1 < v_2 < \cdots < v_n$ be the valuations. The value $v_i$ equals the $i^\text{th}$ order statistic whose distribution is well-known for uniform values. We are interested in the ratio $R=v_{n-k+1} / v_{n-m}$, so let $f(\cdot)$ be its density function. Let $k = m\delta$ for some $0<\delta<1$. By \cref{lemma:ratio_order_stats}, noting that we have $j=(1+\alpha-\delta)m+1, i=\alpha m$, we get that,
    \begin{align*}
        f(r)
        &= \frac{n!\,(r-1)^{m-1}}{(\alpha m - 1)!((1-\delta)m)!\,r^{n}}   \int_{0}^1 (1-u)^{(1+\alpha-\delta)m-1} u^{\delta m - 1} \, \mathrm{d}u\\
        &= \frac{((1+\alpha-\delta)m)!}{((1-\delta)m)!(\alpha m-2)!} \frac{(r-1)^{(1-\delta)m}}{r^n}.
    \end{align*}
    We denote by $H(p) = -p \lg p - (1-p) \lg(1-p)$, the binary entropy function, defined on $[0,1]$. Note that $H(p) \leq 1$ for every $p\in[0,1]$. A useful upper bound is given by the following. \begin{equation}\label{eq:entropy_upper_bound}
        H(x) \leq 2 \sqrt{x\,(1-x)}
    \end{equation}
    The binary entropy function is useful because it allows us to upper bound the binomial coefficient as follows.
    \begin{equation}\label{eq:binom_entropy_bound}
        {n \choose k} \leq 2^{nH(k/n)}
    \end{equation}
    We proceed to bound the probability that \cref{eq:valuation_bound} does not hold as follows.
    \begin{align*}
        &\Pr\left[R' \geq \frac{n-k}{n-m}\right]
        = \int_{\frac{1+\alpha-\delta}{\alpha}}^\infty f_R(r) \,\mathrm{d}r\\
        &= \frac{((1+\alpha-\delta)m)!}{((1-\delta)m)!(\alpha m-2)!}
        \int_{\frac{1+\alpha-\delta}{\alpha}}^\infty\frac{(r-1)^{(1-\delta)m}}{r^n}  \,\mathrm{d}r\\
        &\leq \frac{\alpha}{\alpha+\delta} {{(1+\alpha-\delta)m} \choose {\alpha m}} \left( \frac{1+\alpha-\delta}{\alpha}\right)^{1-(\alpha+\delta)m}\\
        \intertext{We now apply \cref{eq:binom_entropy_bound} and collect the terms in the exponent.}
        &\leq\frac{\alpha}{\alpha+\delta} \exp\Bigg( H\left(\frac{\alpha}{1+\alpha-\delta}\right)(1+\alpha-\delta)\,m \\
        &\hspace{3.3cm} +\log\left(\frac{1+\alpha-\delta}{\alpha}\right)(1-(\alpha+\delta)m) \Bigg)\\
        \intertext{We now use the fact that $H(p)\leq 2\sqrt{p(1-p)}$ as per \cref{eq:entropy_upper_bound} to obtain,}
        &\leq \frac{\alpha}{\alpha+\delta} \exp\Bigg( \log\left(\frac{1+\alpha\delta}{\alpha}\right)\\
        &\quad\quad\quad\quad\quad\quad\quad\quad + m \left[2\sqrt{\alpha(1-\delta)}-(\alpha+\delta) \log\left(\frac{1+\alpha\delta}{\alpha}\right)\right] \Bigg).
    \end{align*}
    We note that the exponent is negative for sufficiently large $m$, and hence the probability negligible if,
    \begin{equation*} \label{eq:sufficient_uniform}
        1 + \alpha - \delta-(\alpha+\delta) \log\left(\frac{1+\alpha\delta}{\alpha}\right)<0.
    \end{equation*}
    Which solves to $0<\alpha < 0.529914$ for $\delta=0.69$. 
\end{proof}

\begin{lemma}[Adler, \cite{pareto_ratio_order}]\label{lemma:pareto_order_stats}
    Let $X_1,X_2, \ldots X_n$ be i.i.d. Pareto distributed with parameter $p>0$. Let $i<j$ and define $R_{ij} = \frac{X_{(j)}}{X_{(i)}}$ and let $f(\cdot)$ be its density function with support $[1,\infty)$. Then for every $r\geq 1$,
    \begin{align*}
        f(r) = \frac{p\,(n-i)!}{(j-i-1)!(n-j)!} \left(1-\frac{1}{r^p}\right)^{j-i-1} \frac{1}{r^{p(n-j+1)+1}}.
    \end{align*} 
\end{lemma}
\begin{thm}
    Suppose $n$ buyers participate in an auction of $m$ identical items where $n=(1+\alpha)\,m$ for some $\alpha>0$. If the valuations of the items are sampled according to a Pareto distribution with parameter $p>1$ and $0\leq \alpha \leq \alpha(p)<0.69$, then first-price auctions are not Stackelberg resilient, except with probability negligible in $m$.
\end{thm}
\begin{proof}
    Suppose for the sake of the argument that $n$ is even, and let $C$ be the $m\delta$ players with the largest valuations for some constant $0<\delta<1$. Let $R = (v_{n-k+1} / v_{n-m})$ and let $f(\cdot)$ be its density function. By \cref{lemma:pareto_order_stats}, it is given by,
        \begin{align*}
            f(r) 
            &= \frac{pm!}{(m - k-2)!(k-1)!} \left(1-\frac{1}{r^p}\right)^{m-k-2}\,\frac1{r^{pk+1}}\\
            &= pk (m-k-1)(m-k){m \choose k} \left(1-\frac{1}{r^p}\right)^{m-k-2}\,\frac1{r^{pk+1}}.
        \end{align*}
        We proceed to bound the probability that \cref{eq:valuation_bound} does not hold as follows.
        \begin{align*}
            &\Pr\left[R \geq \frac{n-k}{n-m}\right] \\
            &= pk (m-k-1)(m-k) {m \choose k} \int_{\frac{1+\alpha-\delta}{\alpha}}^\infty \frac{\left(1-\frac{1}{r^p}\right)^{m-k-2}}{r^{pk+1}} \, \mathrm{d}r \\
            &\leq pk (m-k-1)(m-k) {m \choose k} \int_{\frac{1+\alpha-\delta}{\alpha}}^\infty \frac{1}{r^{pk+1}} \, \mathrm{d}r\\
            &= m((1-\delta)m-1)(1-\delta) {m \choose {\delta m}} \left(\frac{1+\alpha-\delta}{\alpha}\right)^{-p\delta m}\\
        \intertext{We now bound the binomial coefficient using \cref{eq:binom_entropy_bound} and collect the terms in the exponent.}    
            &\leq m((1-\delta)m-1)(1-\delta) \exp\left( m \left[ H(\delta) -  \delta p\log\left(\frac{1+\alpha-\delta}{\alpha}\right) \right] \right)
        \end{align*}
        We note that the exponent is negative, and hence the function negligible, if the following inequality is satisfied.
        \begin{align*}
            \delta p \log\left(\frac{1+\alpha-\delta}{\alpha}\right) > H(\delta).
        \end{align*}
        By \cref{eq:entropy_upper_bound}, it suffices then to establish the following bound.
        \begin{align*}
            \delta p \log\left(\frac{1+\alpha-\delta}{\alpha}\right) > 2 \sqrt{\delta(1-\delta)}.
        \end{align*}
        We now let $\delta = \frac{5}{p^2 + 4}$, and note that this inequality is satisfied for any $p>1$ whenever the following inequality holds.
        $$
            0<\alpha <\frac{p^2-1}{(4+p^2)\left(\exp\left(\frac{2\sqrt{\frac{p^2-1}{p^2}}}{\sqrt{5}}\right) - 1\right)}
        $$
        Denote the rhs by $\alpha(p)$. Note that $\alpha(p)>0$ for any $p>1$ and evaluates to $\frac12(\mathrm{coth}(1/\sqrt{5}) - 1) \approx 0.69$ in the limit as $p \rightarrow \infty$. 
\end{proof}

\section{Conclusion}
In this paper, we demonstrated a Stackelberg attack on auctions and transaction fee mechanisms. Our work suggests that blockchains may be susceptible to these attacks and calls into question the extent to which auctions should be used in these systems, or whether these systems should be supplemented with e.g. lotteries.

\bibliography{ecai}
\end{document}